\documentclass[12pt, draftclsnofoot, onecolumn]{IEEEtran}
\usepackage{amssymb}
\usepackage{amsmath}
\usepackage{cite}
\usepackage{url}
\usepackage{xcolor}
\usepackage{cite,graphicx,amsmath,amssymb}
\usepackage{subfigure}
\usepackage{fancyhdr}
\usepackage{mdwmath}
\usepackage{mdwtab}
\usepackage{caption}
\usepackage{amsthm}
\usepackage{setspace}
\usepackage{xcolor}
\usepackage{url}
\usepackage{subfig}
\usepackage{booktabs}
\usepackage{multirow}
\newtheorem{remark}{Remark}
\newtheorem{theorem}{Theorem}

\newtheorem{lemma}{Lemma}

\newtheorem{corollary}{Corollary}

\newtheorem{proposition}{Proposition}

\begin{document}

\title{Physical Layer Security of Intelligent Reflective Surface Aided NOMA Networks}


\author{
	
	Zhiqing~Tang,
	Tianwei~Hou,~\IEEEmembership{Student Member,~IEEE,}
	Yuanwei~Liu,~\IEEEmembership{Senior Member,~IEEE,}
	Jiankang~Zhang,~\IEEEmembership{Senior Member,~IEEE,}
	and Lajos~Hanzo,~\IEEEmembership{Fellow,~IEEE}
	\thanks{Part of this work has been submitted to IEEE International Conference on Communications (ICC) 2021 \cite{tangICC2021}.}
	\thanks{This work was supported in part by the National Science Foundation of China under Grant 61571401, in part by the National Science Foundation for Young Scientists of China under Grant 61901416.}
	\thanks{Z. Tang is with Zhengzhou University, Zhengzhou,
		China (email:iezqtang@zzu.edu.cn).}
	\thanks{T. Hou is with Beijing Jiaotong University, Beijing,
		China (email:16111019@bjtu.edu.cn).}
	\thanks{Y. Liu is with Queen Mary University of London, London,
		UK (email:yuanwei.liu@qmul.ac.uk).}
	\thanks{J. Zhang and L. Hanzo are with University of Southampton, Southampton, U.K. (email:jz09v@ecs.soton.ac.uk; lh@ecs.soton.ac.uk).}
}

\maketitle
\begin{abstract}
Intelligent reflective surface (IRS) technology is emerging as a promising performance enhancement technique for next-generation wireless networks. 
Hence, we investigate the physical layer security of the downlink in IRS-aided non-orthogonal multiple access networks in the presence of an eavesdropper, where an IRS is deployed for enhancing the quality by assisting the cell-edge user to communicate with the base station.
To characterize the network's performance, the expected value of the new channel statistics is derived for the reflected links in the case of Nakagami-$m$ fading.
Furthermore, the performance of the proposed network is evaluated both in terms of the secrecy outage probability~(SOP) and the average secrecy capacity (ASC).
The closed-form expressions of the SOP and the ASC are derived. We also study the impact of various network parameters on the overall performance of the network considered. To obtain further insights, the secrecy diversity orders and the high signal-to-noise ratio slopes are obtained.
We finally show that: 1) the expectation of the channel gain in the reflected links is determined both by the number of IRSs and by the Nakagami-$m$ fading parameters; 2)~The SOP of both receiver 1 and receiver 2 becomes unity, when the number of IRSs is sufficiently high; 3) The secrecy diversity orders are affected both by the number of IRSs and by the Nakagami-$m$ fading parameters, whereas the high-SNR slopes are not affected by these parameters. 
Our Monte-Carlo simulations perfectly demonstrate the analytical results.
\end{abstract}

\section{Introduction}
In recent years, intelligent reflective surfaces (IRSs\footnote{Also known as Reconfigurable Intelligent Surfaces (RISs).}) have been proposed for beneficially influencing wireless signal propagation\cite{wu2019intelligent,ozdogan2019intelligent,wu2019towards}. 
They can be installed on walls, ceilings, building facades and other infrastructure elements. Secondly, because they may be readily integrated into the existing wireless networks, they constitute a promising cost-effective spectral efficiency~(SE) and/or energy efficiency (EE) improvement technique capable of controlling the scattering, reflection and refraction characteristics of the radio waves\cite{alghamdi2020intelligent,huang2019reconfigurable}. 
However, the design and optimization of IRS-aided networks requires further study. 
By appropriately adjusting the amplitude-reflection and phase coefficients, the signals reflected by the IRS can be superimposed on the direct link for enhancing the received signals \cite{di2019smart,di2019reflection}.
As another benefit, IRSs also have the ability to eliminate undesired signals, such as co-channel interference~\cite{hou2020mimo}.
It was also revealed that IRSs are capable of reducing the outage probability (OP) of optical communication networks~\cite{wang2020performance}.
In \cite{shi2019enhanced}, the authors jointly optimized the active beamforming at the access point~(AP) and the passive beamforming at the IRS of a specific scenario for maximizing the power harvested by an energy harvesting receiver.
In~\cite{wu2019weighted}, the authors studied the weighted sum-power maximization problem of an IRS-aided wireless information and power transfer (SWIPT) network, while the authors of~\cite{gopi2020intelligent} conceived beam index modulation for millimeter wave communication relying on IRSs.

As a parallel development, power domain non-orthogonal multiple access (NOMA\footnote{Throughout this paper, we focus our attention on the family of power-domain NOMA schemes. Hence we simply use ``NOMA" to refer to ``power-domain NOMA" in the following.}) has the capability of providing services for multiple users within the same physical time and/or frequency resource block at the same time, thereby significantly improving the SE and connection density~\cite{liu2017non,dai2015non,ding2017application}. 
In~\cite{li2020outage}, the authors studied the OP of the single cell multi-carrier NOMA downlink, where the transmitter side only has statistical channel state information (CSI) knowledge. 
In~\cite{yuanwei2017TWC}, the authors pointed out that the secrecy diversity order and the asymptotic secrecy outage probability~(SOP) of a pair of NOMA users are determined by that specific user, which has poorer channel gains.
As demonstrated in~\cite{ding2015impact}, to unleash the full potential of NOMA it is important to ensure that an appropriate power difference exists between the users.
The IRS has the ability to change the channel gains for the sake of enhancing the performance of NOMA, hence their intrinsic amalgam has substantial benefits.

Motivated by the potential joint benefits of IRS and NOMA, IRS-aided NOMA networks have been proposed in~\cite{hou2020reconfigurable} for enhancing both the SE and EE.
Furthermore, both the downlink~(DL) and uplink of IRS-aided NOMA networks were studied~\cite{yanyu2020}.
Additionally, IRS-aided NOMA transmission was contrasted to spatial division multiple access (SDMA) in~\cite{ding2020simple}. 
The performance of IRS-aided NOMA networks relying on both idealized perfect and imperfect successive interference cancellation was investigated in~\cite{yue2020performance}.
In~\cite{mu2019exploiting}, the authors considered both ideal and realistic non-ideal IRS assumptions by jointly optimizing the active beamforming at the BS and the passive beamforming at the IRS of a link to maximize the sum rate.
Furthermore, in \cite{fu2019reconfigurable}, the authors jointly optimized the active beamforming at the BS and the passive beamforming at the IRS of a link for minimizing the total transmit power.
In~\cite{zhang2020downlink}, the authors derived the closed-form expression of the coverage probability of appropriately paired NOMA users and proved that the performance of the IRS-aided NOMA network is better than that of the traditional NOMA network.
A new type of multi-cell IRS-aided NOMA resource allocation framework was proposed in~\cite{ni2020resource}, while an IRS-aided millimeter wave (mmWave) NOMA network was considered in~\cite{zuo2020intelligent}.
In \cite{liu2020ris}, machine learning techniques were adopted in an IRS-aided NOMA network for maximizing the EE. 

Given the broadcast nature of wireless transmission, the issue of physical layer security (PLS) has attracted widespread interests~\cite{wang2016physical,wu2018survey}. Although rigorous efforts have been dedicated to the PLS of wireless communications, the overall progress has remained relatively slow~\cite{almohamad2020smart}. 
However, the emergence of IRS technology provides a new horizon for PLS problems.
In~\cite{yang2020secrecy}, the authors studied the secrecy performance of an IRS-aided wireless communication network in the presence of an eavesdropper (Eve). Similarly, in~\cite{cui2019secure}, the authors investigated an IRS-aided secure wireless communication network, where the eavesdropping channels are stronger than the legitimate communication channels.
In~\cite{yang2020secrecySOP}, the authors studied the SOP of IRS-aided NOMA networks in a multi-user scenario. 
In~\cite{8972400}, the authors investigated whether the use of artificial noise is helpful in terms of enhancing the secrecy rate of IRS-aided networks.
Additionally, a mobile wiretap network was also investigated~\cite{long2020reflections}, in which an unmanned aerial vehicle (UAV) equipped with IRS was considered.
The PLS of an IRS-aided vehicular network was studied in~\cite{makarfi2020reconfigurable}.

\begin{table}[htbp]
	\caption{ Conclusion of the current literature, and the objective of the proposed model.}
	\begin{center}
		\centering
		\begin{tabular}{|l|c|c|c|c|c|c|c|}
			\hline
			\centering
			& \cite{huang2019reconfigurable} & ~\cite{mu2019exploiting}   &  \cite{yang2020secrecy} & \cite{yang2020secrecySOP} & \cite{makarfi2020reconfigurable} &  \cite{dong2020secure}& Proposed scheme \\
			\hline
			\centering
			IRS-aided network             & \checkmark & \checkmark  & \checkmark & \checkmark & \checkmark & \checkmark & \checkmark \\
			\hline
			\centering
			OMA scheme                    & \checkmark & \checkmark & \checkmark &            &  \checkmark &            & \checkmark \\
			\hline
			\centering
			NOMA scheme                   &           &  \checkmark &           & \checkmark &              & \checkmark & \checkmark \\
			\hline
			\centering
			Approximated distribution     &            &            & \checkmark & \checkmark &             &           & \checkmark \\
			\hline
			\centering
			Exact distribution            &            &            &            &            &             &            & \checkmark \\
			\hline
			\centering
			SOP analysis                  &            &            & \checkmark & \checkmark & \checkmark   &           & \checkmark \\
			\hline
			\centering
			Secrecy diversity order       &            &           &            &            &             &            & \checkmark \\
			\hline
			\centering
			ASC analysis                  &            &            &            &           &  \checkmark &            & \checkmark \\
			\hline
			\centering
			High-SNR slpoe                &            &            &            &            &             &           & \checkmark \\
			\hline
		\end{tabular}
	\end{center}
	\label{table:RIS_compare_benchmarks}
\end{table}
\subsection{Motivation and Contribution} 
Most of the existing contributions on PLS in IRS-aided networks studied the problem of maximizing the secrecy rate. An IRS-aided  multiple-input single-output and multiple-input multiple-output network has been studied in \cite{dong2020secure}. 
Furthermore, in \cite{8847342,8743496}, IRS-aided multiple-input single-output networks have been studied. 
As mentioned above, PLS has been studied in diverse scenarios, but
there is a paucity of investigations on IRS-aided NOMA, which motivates this contribution. 
Table~\ref{table:RIS_compare_benchmarks} boldly and explicitly
contrasts our schemes to the pertinent literature. 
Specifically, we consider the scenario of an IRS-aided NOMA network, where a BS supports a cell-center user as well as a cell-edge user, and NOMA is invoked. Additionally, an Eve is close to the cell-edge users. Explicitly, we consider a NOMA network in which the cell-edge user cannot directly communicate with the BS, hence relies on the IRS to communicate with the BS. Our new contributions can be summarized in more detail as follows:

\begin{itemize}
\item  We investigate the secrecy performance of IRS-aided NOMA networks, where BS communicates with a pair of NOMA users in the presence of an Eve.
In particular, we investigate the SOP and the average secrecy capacity (ASC).
\item We adopt the Nakagami-$m$ fading model for the reflected links so that it can be either Line-of-sight (LoS) or non-line-of-sight (NLoS). Correspondingly, we derive new channel statistics for the reflected links based on the associated Laplace transforms and moments. We demonstrate that the expectation of the channel gain for the reflected links is determined by the number of IRSs and the Nakagami-$m$ fading parameters.
\item We derive closed-form expressions of the SOP and the ASC for the proposed network. 
To glean further insights, we derive asymptotic approximations of the SOP and the ASC in the high signal-to-noise ratio (SNR) regime to derive both the secrecy diversity order and the high-SNR slope. We demonstrate that the number of IRSs and Nakagami-$m$ fading parameters directly affect the secrecy diversity order, but have no effect on the high-SNR slope.
\end{itemize}

\subsection{Organization}
The rest of the paper is organized as follows. In Section \ref{SYSTEM MODEL}, the model of the IRS-aided NOMA network is discussed. In Section \ref{Performance and analysis}, the new channel statistics of the reflected links are derived, and then the performance analysis of IRS-aided NOMA networks is conducted. Furthermore, our numerical and simulation results are composed in Section \ref{NR}. Finally, our conclusions are offered in Section \ref{CONCLUSIONS}.

In this paper, scalars are denoted by italic letters. For a scalar $s$, $s!$ denotes the factorial of $s$. Vectors and matrices are denoted by boldface letters. For a vector $\mathbf{v}$, $\mathbf{v}^T$ denotes the transpose of $\mathbf{v}$, and diag($\mathbf{v}$) denotes a diagonal matrix in which each diagonal element is the corresponding element in $\mathbf{v}$, respectively. $\mathbb{P}(\cdot)$ denotes the probability, and $\mathbb{E}(\cdot)$ represents the expectation.

\section{System Model}\label{SYSTEM MODEL}

As shown in Fig. 1, we consider the secure DL of an IRS-aided NOMA network, where a BS communicates with two legitimate users in the presence of an Eve. For the pair of legitimate users, the NOMA transmission protocol is invoked. It is assumed that the BS, the paired NOMA users and Eve are equipped with a single antenna.
We also assumed that Eve has powerful detection capability, which is capable of detecting the messages of the paired NOMA users.
We have an IRS at the appropriate location, who has $N$ intelligent surfaces. 
More specifically, user 2 is the cell-edge user, who needs help from the IRS to communicate with the BS. At the same time, user 1 is the cell-center user who can directly communicate with the BS.
Furthermore, there is no direct link between the IRS and user 1, as well as between the BS supporting user 2 and Eve, due to their long-distance and the presence of blocking objects.
\begin{figure} [t!]
	\centering
	\includegraphics[width= 4in]{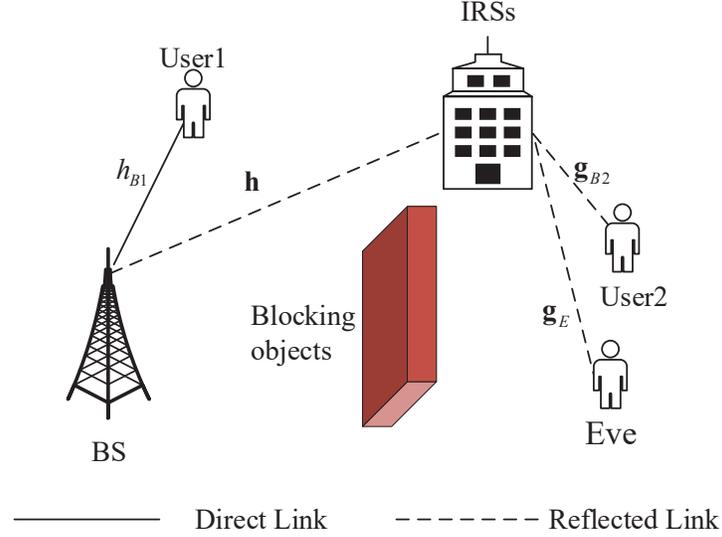}
	\caption{An IRS-aided NOMA secure communication.}
	\label{system_model}
\end{figure}

The small-scale fading vector between the BS and IRSs is denoted by
\begin{equation}
\mathbf{h}=\left[h_1,h_2,\cdots,h_N\right]^T.
\end{equation}

The small-scale fading vectors between the IRSs and user 2 and that between the IRSs and Eve are given by
\begin{equation}
\mathbf{g}_{B2}=\left[g_{2,1},g_{2,2},\cdots,g_{2,N}\right],
\end{equation}
and
\begin{equation}
\mathbf{g}_E=\left[g_{E,1},g_{E,2},\cdots,g_{E,N}\right],
\end{equation}
respectively. The elements in $\mathbf{h}$, $\mathbf{g}_{B2}$ and $\mathbf{g}_E$ obey the Nakagami-$m$ distribution having fading parameters of $m_1$, $m_2$, and $m_3$, respectively.
Moreover, the direct link between the BS and user 1 is modelled by Rayleigh fading, which is denoted by $h_{B1}$.

The BS sends the following signal to the paired NOMA users
\begin{equation}
s=\sqrt{a_1}s_1+\sqrt{a_2}s_2,
\end{equation}
where $s_1$ and $s_2$ are the signal intended for user 1 and user 2, respectively, while $\sqrt{a_1}$ and $\sqrt{a_2}$ are the power allocation factors of user 1 and user 2, respectively. Based on the NOMA protocol, the power allocation coefficients satisfy the condition that $a_1+a_2=1$.

The signal received by user 1 and user 2 are given by
\begin{equation}
y_1=h_{B1}\sqrt{d_{B1}^{-\alpha_{B1}}}Ps+n_1,
\end{equation}
and
\begin{equation}
y_2=\mathbf{g}_{B2}\mathbf{\Phi}\mathbf{h}\sqrt{d_1^{-\alpha_1}d_{B2}^{-\alpha_{B2}}}Ps+n_2,
\end{equation}
respectively,
where $d_{B1}$ and $d_1$ denote the distances from the BS to user 1 and the IRS, respectively. Furthermore, $d_{B2}$ denotes the distance from the IRS to user 2, $\alpha_{B1}$, $\alpha_1$ and $\alpha_{B2}$ represent the path loss exponents of the BS-user 1 link, BS-IRS links and IRS-user 2 links, respectively. 
Still referring to $P$ denotes the transmit power of the BS, $n_1$ and $n_2$ represent additive white Gaussian noise (AWGN) at user 1 and user 2, respectively, which is modeled as a realization of a zero-mean complex circularly symmetric Gaussian variable with variance $\sigma^2$. Additionally, $\mathbf{\Phi}\triangleq \text{diag} [\beta_1\phi_1,\beta_2\phi_2,\dots,\beta_N\phi_N]$ is a diagonal matrix, which represents the effective phase shift applied by all IRSs, where $\beta_n\in(0,1]$ is the amplitude reflection coefficient of the IRSs, while $\phi_n=\exp(j\theta_n)$, $j=\sqrt{-1}$, $\forall{n}=1,2,\dots,N$ and $\theta_n\in[0,2\pi)$, represents the phase shift caused by the $n$-th IRS. 
The signal received by Eve is given by
\begin{equation}
y_E=\mathbf{g}_E\mathbf{\Phi}\mathbf{h}\sqrt{d_1^{-\alpha_1}d_{E}^{-\alpha_{E}}}Ps+n_E,
\end{equation}
 where $d_E$ is the distance form the IRS to Eve, $\alpha_{E}$ denotes the path loss exponent of the IRS-EVE link and $n_E$ represents the AWGN at Eve.

\section{Performance and analysis} \label{Performance and analysis}
In this section, we consider an IRS design similar to \cite{yanyu2020}. Specifically, in order to simultaneously control multiple IRSs, the CSI of the paired NOMA users' channels is assumed to be perfectly known. However, the CSI of Eve is not available. We first derive new channel statistics for both the direct link and reflected links. Then the SOP and the ASC are illustrated in the following subsections.

\subsection{New Channel Statistics}
According to the previous assumption, the instantaneous signal-to-noise ratio (SNR) of user 1 and the instantaneous signal-to-interference-plus-noise ratio (SINR) of user 2 can be expressed as
\begin{equation}\label{gammaB1}
\gamma_{B1}=\rho a_1 |{h}_1|^2d_{B1}^{-\alpha_{B1}},
\end{equation}
and
\begin{equation}\label{gammaB2}
\gamma_{B2}=\frac{a_2|\hat{h}_{B2}|^2d_1^{-\alpha_1}d_{B2}^{-\alpha_{B2}}}{a_1|\hat{h}_{B2}|^2d_1^{-\alpha_1}d_{B2}^{-\alpha_{B2}}+\frac{1}{\rho}},
\end{equation}
respectively, where $\rho$ denotes the transmit SNR and $\hat{h}_{B2}=\sum_{n=1}^{N}|g_{2,n}||h_{n}|$ represents the equivalent channel of the BS-IRS-user 2 links.

The phase shifts are designed for user 2, hence the effective channel gain of Eve cannot be evaluated.
In this paper, we consider the worst-case scenario of the IRS-aided NOMA network, in which all of the BS-IRS-Eve signals are co-phased.
Therefore, the equivalent channel of Eve is similar to that of user 2, which can be expressed as
\begin{equation}
\hat{h}_E=\sum_{n=1}^{N}|g_{E,n}||h_{n}|.
\end{equation}
Therefore, the instantaneous SNR of detecting the information of user 1 and user 2 can be expressed as:
\begin{equation}\label{gammaE}
\gamma_{E_i}=\rho_ea_i|\hat{h}_E|^2d_1^{-\alpha_1}d_{E}^{-\alpha_{E}},
\end{equation}
where $i\in\{1,2\}$.

\begin{lemma}
	The cumulative distribution function (CDF) of $\gamma_{B1}$ is
\begin{equation}\label{CDFB1}
F_{\gamma_{B1}}(x) = 1-e^{-\frac{x}{a_1\rho d_{B_1}^{-\alpha_{B1}}}} ,
\end{equation}
\end{lemma}
\begin{proof}
	the CDF of $|h_{B1}|^2$ is $F_{|h_{B1}|^2}(x)=1-e^{-x}$. Then, $F_{\gamma_{B1}}$ can be derived as
	\begin{equation}
	\begin{aligned}
	F_{\gamma_{B1}}(x)&=\mathbb{P}\left(\gamma_{B1}<x\right)\\
	&=F_{|h_{B1}|^2}\left(\frac{x}{a_1\rho d_{B_1}^{-\alpha_{B1}}}\right)\\
	&=1-e^{-\frac{x}{a_1\rho d_{B_1}^{-\alpha_{B1}}}} .
	\end{aligned}
	\end{equation}
	 This completes the proof.
\end{proof}
\begin{lemma}\label{lemma 2}
	Recall that the fading parameters of the elements in $\mathbf{h}$ and $\mathbf{g}_{B2}$ are $m_1$ and $m_2$, respectively. The CDF of $\gamma_{B2}$ in the low-SNR regimes and the high-SNR regimes (when $m_1\neq m_2$) are given by
	\begin{equation}\label{cdflowb2}
	F_{\gamma_{B2}}(x)=1-Q_{\frac{1}{2}}\left(\sqrt{\lambda},\sqrt{\frac{x}{N(1-\epsilon)\rho(a_2-a_1x)d_1^{-\alpha_1}d_{B2}^{-\alpha_{B2}}}}\right),
	\end{equation}
and
	\begin{equation}\label{cdfhighb2}
	F_{\gamma_{B2}}^{0_+}(x)=\frac{{m}^N(4m_sm_l)^{-m_sN}}{\Gamma(2m_sN)}\gamma\left( 2m_sN,\frac{2\sqrt{m_sm_lx}}{\sqrt{\rho(a_2-a_1x)d_1^{-\alpha_1}d_{B2}^{-\alpha_{B2}}}}\right),
	\end{equation}
	respectively, where $x<\frac{a_2}{a_1}$, $Q_{\alpha}(\cdot,\cdot)$ is the Marcum Q-function, $\epsilon=\frac{1}{m_1m_2}\left(\frac{\Gamma(m_1+\frac{1}{2})}{\Gamma(m_1)}\right)^2\left(\frac{\Gamma(m_2+\frac{1}{2})}{\Gamma(m_2)}\right)^2$, $\lambda=\frac{N\epsilon}{1-\epsilon}$, ${m}=\frac{\sqrt{\pi}4^{m_s-m_l+1}(m_sm_l)^m_s\Gamma(2m_s)\Gamma(2m_l-2m_s)}{\Gamma(m_s)\Gamma(m_l)\Gamma(m_l-m_s+\frac{1}{2})}$,  $m_l=\max\left(m_1,m_2\right)$, $m_s=\min\left(m_1,m_2\right)$, $\Gamma(\cdot)$~denotes the Gamma function and $\gamma(\cdot,\cdot)$ is the lower incomplete Gamma function.
\end{lemma}
\begin{proof}
	Please refer to Appendix A.
\end{proof}

\begin{lemma}\label{lemma 3}
	Upon introducing $Z=\sum_{n=1}^{N}|g_{q,n}||h_n|$, where $q\in\{2,E\}$, the expectation of $Z^2$ is given by	
	\begin{equation}\label{E^2}
	\begin{aligned}
	\mathbb{E}(Z^2)
	&=aN\omega^Nd^{-aN}k_1^N\mho ,
	\end{aligned}
	\end{equation}
	where $a=2m_c$, $b=m_c-m_d+\frac{1}{2}$, $c=m_c+m_d+\frac{1}{2}$, $d=2\sqrt{m_cm_d}$, $\omega=\frac{\sqrt{\pi}4^{m_c-m_d+1}(m_cm_d)^{m_c}\Gamma(2m_c)\Gamma(2m_d)}{\Gamma(m_c)\Gamma(m_d)\Gamma\left(m_c+m_d+\frac{1}{2}\right)}$ and $\mho=\mho_1+\mho_2-\mho_3+\mho_4$ with $\mho_1=\frac{aN+1}{d^2}$, $\mho_2=\frac{4ab^2k_2^2}{c^2d^2k_1^2}(N-1)$, $\mho_3=\frac{4abNk_2}{cd^2k_1}$ and $\mho_4=\frac{4(a+1)(b^2+b)}{(c^2+c)d^2}k_3-\frac{4b}{cd^2}k_2$. Furthermore, $k_1={}_{2}F_{1}\left(a,b;c;-1\right)$, $k_2={}_{2}F_{1}\left(a+1,b+1;c+1;-1\right)$ and $k_3={}_{2}F_{1}\left(a+2,b+2;c+2;-1\right)$ are the Gauss hypergeometric function\cite[\textrm{eq}. (9.100)]{gradshteyn2007}. In the case of $q=2$, we have $m_c=\min(m_1,m_2)$,  $m_d=\max(m_1,m_2)$ and $\mathbb{E}(Z^2)$ denotes the average channel gain of the BS-IRS-user 2 links, otherwise, we have $m_c=\min(m_1,m_3)$, $m_d=\max(m_1,m_3)$ and $\mathbb{E}(Z^2)$ denotes the average channel gain of the BS-IRS-Eve links.

\end{lemma}
\begin{proof}
		Please refer to Appendix B.
\end{proof}
\begin{remark}
	The second moment $\mathbb{E}(Z^2)$ is accurate for the global CSIs. 
	Hence, we observe from~\eqref{E^2} that the expectation of the channel gain for the reflected links is determined by the number of IRSs and the Nakagami-$m$ fading parameters.
\end{remark}

\subsection{Secrecy Outage Probability}\label{section 4}
In the proposed network, the capacity of the main channel for the $i$-th user is given by $C_{B_i}=\log(1+\gamma_{B_i})$, while the capacity of Eve's channel for the $i$-th user is quantified by $C_{E_i}=\log(1+\gamma_{E_i})$. 
As such, the secrecy rate of the $i$-th user can be expressed as
\begin{equation}\label{secrecy capacity}
C_i=\left[C_{B_i}-C_{E_i}\right]^{+},
\end{equation}
where $[x]^+=\max\{x,0\}$.

\subsubsection{SOP analysis}~

Next, we focus our attention on the SOP of user 1 and user 2. 
For the $i$-th user, we note that, if $C_i<R_i$, information transmission at a rate of $R_i$ is compromised. The SOP of the $i$-th user can be expressed as
\begin{equation}\label{SOP}
\begin{aligned}
P_i(R_{i})&=\mathbb{P}\left(C < R_{i} \right)\\
&=\mathbb{P}\left(\log_2\left(\frac{1+\gamma_{Bi}}{1+\gamma_{Ei}}\right)< R_{i} \right)\\
&=\mathbb{P}\left(\gamma_{Bi}<2^{R_{i}}\left(1+\gamma_{Ei}\right)-1\right).
\end{aligned}
\end{equation}
Then we derive the SOP of user 1 and user 2 in the following theorems. 
\begin{theorem}\label{theorem SOP1}
	In the IRS-aided NOMA network considered, the SOP of user 1 is given by
	\begin{equation}\label{SOP1}
     P_1(R_1) \approx 1-e^{\frac{-y_1}{a_1\rho d_{B_1}^{-\alpha_{B_1}}}} ,
	\end{equation}
	where $y_1=2^{R_1}\left(1+a_1\rho_e\mu d_1^{-\alpha_1}d_{E}^{-\alpha_{E}}\right)-1$, and $\mu=\mathbb{E}\left(|\hat{h}_E|^2\right)$ which can be obtained from \textbf{Lemma 3} in the case of $m_c=\min(m_1,m_3)$ and $m_d=\max(m_1,m_3)$.
\end{theorem}
\begin{proof}
	Based on \eqref{gammaE} and \eqref{SOP}, we have
\begin{equation}\label{P1TD}
\begin{aligned}
P_1(R_1) &= \mathbb{P}\left(\gamma_{B_1}<2^{R_{1}}\left(1+\gamma_{E_1}\right)-1\right) \\
&\approx \mathbb{P}\left(\gamma_{B_1}<2^{R_{1}}\left(1+|\hat{h}_E|^2d_1^{-\alpha_1}d_{E}^{-\alpha_{E}}\right)-1\right) \\
&= F_{\gamma_{B1}}\left(2^{R_{1}}\left(1+\rho_ea_1\mathbb{E}\left(|\hat{h}_E|^2\right)d_1^{-\alpha_1}d_{E}^{-\alpha_{E}}\right)-1\right).
\end{aligned}
\end{equation}

 Then, by substituting \eqref{CDFB1} and \eqref{E^2} into \eqref{P1TD}, in the case of $m_c=\min(m_1,m_3)$ and $m_d=\max(m_1,m_3)$, we arrive at \eqref{SOP1} after some further manipulations. This completes the proof.	
\end{proof}

\begin{theorem}\label{theorem SOP2}
	The SOP of user 2 in the low-SNR and high-SNR regimes (when $m_1\neq m_2$) are given by
	\begin{equation}\label{SOP2 low-SNR}
	\begin{aligned}
	P_2^l(R_2)
	&\approx e^{-\frac{\lambda}{2}}\sum_{k=0}^{\infty}\frac{\lambda^k\gamma\left(k+\frac{1}{2},\frac{y_l}{2}\right)}{k!2^k\Gamma\left(k+\frac{1}{2}\right)},
	\end{aligned}
	\end{equation}
	and
	\begin{equation}\label{SOP2-high-SNR}
	P_2^h(R_2)\approx \frac{{m}^N(4m_sm_l)^{-m_sN}}{\Gamma(2m_sN)}\gamma\left( 2m_sN,2\sqrt{m_sm_l}y_h\right),
	\end{equation}
	respectively, 
	where we have
	\begin{equation}
	y_l=\frac{y_2}{N(1-\epsilon)\rho(a_2-a_1y_2)d_1^{-\alpha_1}d_{B2}^{-\alpha_{B2}}},
	\end{equation}
	\begin{equation}
	y_h=\frac{\sqrt{y_2}}{\sqrt{\rho(a_2-a_1y_2)d_1^{-\alpha_1}d_{B2}^{-\alpha_{B2}}}},
	\end{equation}
	and
	\begin{equation}
	y_2=2^{R_{2}}\left(1+a_2\mu\rho_ed_1^{-\alpha_1}d_{E}^{-\alpha_{E}}\right)-1.
	\end{equation}
\end{theorem}
\begin{proof}
	Based on \eqref{gammaE}, \eqref{SOP} and \textbf{Lemma 2}, the SOP of user 2 in the low-SNR regime can be expressed as
	\begin{equation}\label{P2TDL}
	\begin{aligned}
	P_2^l(R_2) &= \mathbb{P}\left(\gamma_{B_2}<2^{R_{2}}\left(1+\gamma_{E_2}\right)-1\right) \\
	& \approx F_{\gamma_{B2}}\left(2^{R_{2}}\left(1+a_2\rho_e\mathbb{E}\left(|\hat{h}_E|^2\right)d_1^{-\alpha_1}d_{E}^{-\alpha_{E}}\right)-1\right).
	\end{aligned}
	\end{equation}
	
Then, by substituting \eqref{cdflowb2} and \eqref{E^2} into \eqref{P2TDL}, in the case of $m_c=\min(m_1,m_3)$ and $m_d=\max(m_1,m_3)$, \eqref{SOP2 low-SNR} can be obtained after some further mathematical manipulations.
	
The SOP of user 2 in the high-SNR regime can be expressed as
	\begin{equation}\label{P2TDH}
	\begin{aligned}
	P_2^h(R_2)&= \mathbb{P}\left(\gamma_{B_2}<2^{R_{2}}\left(1+\gamma_{E_2}\right)-1\right) \\
	&\approx F_{\gamma_{B2}}^{0_+}\left(2^{R_{2}}\left(1+a_2\rho_e\mathbb{E}\left(|\hat{h}_E|^2\right)d_1^{-\alpha_1}d_{E}^{-\alpha_{E}}\right)-1\right).
	\end{aligned}
	\end{equation}

Then, by substituting \eqref{cdfhighb2} and \eqref{E^2} into \eqref{P2TDH}, in the case of $m_c=\min(m_1,m_3)$ and $m_d=\max(m_1,m_3)$, \eqref{SOP2-high-SNR} can be obtained after some further mathematical manipulations. This completes the proof.
\end{proof}
\begin{proposition}
	Both the SOP of user 1 and the SOP of user 2 are $1$ when the number of IRSs is sufficiently high.
\end{proposition}
\begin{proof}
	By substituting $N\to\infty$ into \eqref{E^2}, we have
	\begin{equation}\label{mu}
	\begin{aligned}
	\mu&\approx a^2N^2\omega^Nd^{-aN-2}k_1^N\left(1-\frac{2bk_2}{ck_1}\right)^2\\
	&\approx a^2N^2 \left(\frac{\omega k_1}{d}\right)^N\left(1-\frac{2bk_2}{ck_1}\right)^2.
	\end{aligned}
	\end{equation}
	Since we have $c=1+a-b$, $k_1$ can be rewritten as $k_1=\frac{\Gamma(m_c+m_d+\frac{1}{2})\Gamma(1+m_c)}{\Gamma(1+2m_c)\Gamma(m_d+\frac{1}{2})}$.
Let $\mu_1=\frac{\omega k_1}{d}$, then by substituting $k_1$ into $\mu_1$, we have
\begin{equation}\label{mu1}
\begin{aligned}
\mu_1&=\frac{\sqrt{\pi}4^{m_c-m_d+1}(m_cm_d)^{m_c}\Gamma(2m_c)\Gamma(2m_d)}{\Gamma(m_c)\Gamma(m_d)\Gamma\left(m_c+m_d+\frac{1}{2}\right)}\frac{\Gamma(m_c+m_d+\frac{1}{2})\Gamma(1+m_c)}{2\sqrt{m_cm_d}\Gamma(1+2m_c)\Gamma(m_d+\frac{1}{2})}\\
&=\frac{\sqrt{\pi}4^{m_c-m_d+1}(m_cm_d)^{m_c}\Gamma(2m_c)\Gamma(2m_d)\Gamma(1+m_c)}{\Gamma(m_c)2\sqrt{m_cm_d}\Gamma(1+2m_c)\Gamma(m_d+\frac{1}{2})}.
\end{aligned}
\end{equation}
Given that $\Gamma(1+x)=x\Gamma(x)$ and $\Gamma(x)\Gamma(x+\frac{1}{2})=2^{1-2x}\sqrt{\pi}\Gamma(2x)$, \eqref{mu1} can be rewritten as
\begin{equation}
\mu_1=2^{2m_c-1}(m_cm_d)^{m_c-\frac{1}{2}}\Gamma(m_c)\Gamma(m_d). 
\end{equation}
Since $m_d> m_c\geq\frac{1}{2}$, we have $\mu_1>1$. Then by substituting $\mu_1$ into \eqref{mu}, we have $\mu\to\infty$ when $N\to\infty$.
By  substituting $\mu$ in to \textbf{Theorem 1} and \textbf{Theorem 2}, in the case of $N\to\infty$, we have $P_1(R_1)=1$ and $P_2(R_2)=1$. This completes the proof.	
\end{proof}


\subsubsection{Asymptotic SOP and Secrecy Diversity Order Analysis} ~

In order to derive the secrecy diversity order to gain further insights into the network's operation in the high-SNR regime, the asymptotic behavior is analyzed. 
Again, as the worst-case scenario, we assume that Eve has a powerful detection capability, hence all of the reflected signals are co-phased. 
Without loss of generality, it is assumed that the transmit SNR for the paired NOMA users is sufficiently high (i.e., $\rho\to \infty$), and the SNR of the BS-IRS-Eve links is set to arbitrary values.
Please note that the SOP tends to $1$, when Eve's transmit SNR $\rho_e\to\infty$.
The secrecy diversity order can be defined as follows:
\begin{equation}\label{secrecy diversity order}
d_s=-\lim\limits_{\rho\to\infty}\frac{\log P^{\infty}}{\log \rho},
\end{equation}
where $P^{\infty}$ is the asymptotic SOP.
\begin{corollary}\label{CASOP1}
	The asymptotic SOP of user 1 is given by
	\begin{equation}\label{asymptotic SOP Bob1}
	P_1^{\infty}(R_1)=\frac{y_1}{a_1\rho d_{B_1}^{-\alpha_{B_1}}}.
	\end{equation}
\end{corollary}
\begin{proof}
	By expanding the exponential function in \eqref{SOP1} and extracting the leading-order term, \eqref{asymptotic SOP Bob1} is obtained. This completes the proof.
\end{proof}
\begin{remark}\label{sdo1}
	Upon substituting \eqref{asymptotic SOP Bob1} into \eqref{secrecy diversity order}, the secrecy diversity order of user 1 is $1$.
\end{remark}
\begin{proposition}\label{SOP1 floor}
	The floor of $P_1(R_1)$ in the case of $\rho_e=\rho$ is given by
	\begin{equation}\label{SOP1 lower bound}
	P_{1,\infty}^{\infty}(R_1)=\frac{2^{R_1}\mu d_1^{-\alpha_1}d_E^{-\alpha_{E}}}{d_{B_1}^{-\alpha_{B_1}}}.
	\end{equation}
\end{proposition}
\begin{proof}
	By Substituting $\rho_e=\rho$ into \eqref{asymptotic SOP Bob1}, after some mathematical manipulations, \eqref{SOP1 lower bound} can be obtained. This completes the proof.
\end{proof}

\begin{corollary}\label{CASOP2}
	The asymptotic SOP of user 2 is given by
	\begin{equation}\label{asymptotic SOP2}
	P_2^{\infty}(R_2)=\frac{m^Ny_h^{2m_sN}}{\Gamma(2m_sN+1)}.
	\end{equation}
\end{corollary}
\begin{proof}
	Based on \textbf{Theorem 2}, we have the SOP of user 2 in the high-SNR regime.
	Then, by using the expansions of the lower incomplete Gamma function \cite[\text{eq. (8.354.1)}]{gradshteyn2007}, \eqref{SOP2-high-SNR} can be represented as
	\begin{equation}\label{lower incomplete Gamma function}
	\begin{aligned}
	P_2^{\infty}(R_2)&=\frac{{m}^N(4m_sm_l)^{-m_sN}}{\Gamma(2m_sN)}\sum_{k=0}^{\infty}\frac{(-1)^k\left(2\sqrt{m_sm_l}y_h\right)^{2m_sN+k}}{k!\left(2m_sN+k\right)}.
	\end{aligned}
	\end{equation}
	By extracting the leading-order term in \eqref{lower incomplete Gamma function}, \eqref{asymptotic SOP2} can be obtained.
	This completes the proof.
\end{proof}

\begin{remark}\label{sdo2}
	Upon substituting \eqref{asymptotic SOP2} into \eqref{secrecy diversity order}, the secrecy diversity order of user 2 is $m_sN$.
\end{remark}
\begin{remark}
The secrecy diversity order of user 1 is not affected by the number of IRSs and by the Nakagami-$m$ fading parameters. By contrast, the secrecy diversity order of user 2 is affected by the number of IRSs and by the Nakagami-$m$ fading parameters.
\end{remark}

In this paper, based on the assumptions of using perfect SIC for the paired NOMA users and strong detection capability of Eve, the secrecy outage occurrences of user 1 and user 2 are independent. As a consequence, we define the SOP for the network as that of either the outage of user 1 or of user 2.
\begin{equation}\label{systemSOP}
P_{1,2}=1-(1-P_1)(1-P_2),
\end{equation}
where $P_1$ and $P_2$ are given by \textbf{Theorem \ref{theorem SOP1}} and \textbf{Theorem \ref{theorem SOP2}}, respectively.
\begin{proposition}
	The secrecy diversity order of the network can be expressed as
	\begin{equation}\label{SDO}
	d_s=\min(1,m_sN).
	\end{equation}
\end{proposition}
\begin{proof}
	Based on \textbf{Corollary \ref{CASOP1}} as well as \textbf{Corollary \ref{CASOP2}}, and upon substituting \eqref{asymptotic SOP Bob1} and \eqref{asymptotic SOP2} into \eqref{systemSOP}, the asymptotic SOP for the network can be expressed as
	\begin{equation}\label{system SOP}
	\begin{aligned}
	P_{1,2}^{\infty}&=P_1^{\infty}+P_2^{\infty}-P_1^{\infty}P_2^{\infty}\\
	&\approx \begin{cases}
	\frac{m^N}{\Gamma(2m_sN+1)}\left(\frac{y_2}{\rho(a_2-a_1y_2)d_1^{-\alpha_1}d_{B2}^{-\alpha_{B2}}}\right)^{m_sN}, \quad &1>m_sN\\
	\frac{y_1}{a_1\rho d_{B_1}^{-\alpha_{B_1}}}, \qquad &1<m_sN
	\end{cases}.
	\end{aligned}
	\end{equation}

By substituting \eqref{system SOP} into \eqref{secrecy diversity order}, we arrive at \eqref{SDO}. The proof is completed.
\end{proof}
\begin{remark}\label{remark 6}
	The secrecy diversity order of the network is determined by the smallest between $1$ and $m_sN$.
\end{remark}
\textbf{Remark \ref{remark 6}} provides insightful guidelines for improving the SOP of the network considered by invoking IRS-aided NOMA. 
The SOP of the network is determined by user 1 in the case of $N\geq2$, because the Nakagami-$m$ fading parameters obey $m\geq\frac{1}{2}$. 
Specifically, the SOP of the network is determined by user 1, when the links between the BS and IRS as well as that between the IRS and user 2 is Rayleigh or 
Rician fading.

\subsection{Average Secrecy Capacity Analysis}\label{section 5}

\subsubsection{Approximate ASC}~

In this subsection, we derive analytical expressions for the ASC of the network. The secrecy capacity in \eqref{secrecy capacity} can be rewritten as
\begin{equation}\label{SC}
C_i=
\begin{cases}
\log_2\left(1+\gamma_{B_i}\right)-\log_2\left(1+\gamma_{E_i}\right), \quad &\gamma_{B_i}>\gamma_{E_i}\\
0, &\gamma_{B_i}\leq\gamma_{E_i}
\end{cases}.
\end{equation}
The approximate ASC can be expressed as
\begin{equation}\label{approximate ASC}
C_i\approx R_{B_i}-R_{E_i},
\end{equation}
where $R_{B_i}$ and $R_{E_i}$ are the ergodic rates of the paired NOMA users and Eve, respectively, which can be expressed as
\begin{equation}\label{ERB}
R_{B_i}=\mathbb{E}\big[\log_2\left(1+\gamma_{B_i}\right)\big],
\end{equation}
and
\begin{equation}\label{ERE}
R_{E_i}=\mathbb{E}\big[\log_2\left(1+\gamma_{E_i}\right)\big],
\end{equation}
respectively. 

The ergodic rates of the paired NONA users have been studied in \cite{yanyu2020}, and that of Eve is necessary for obtaining the ASC.
Hence, the ergodic rate of Eve and the approximate ASC are presented in the following theorems.

\begin{theorem}\label{theorem 4}
	In the IRS-aided NOMA network considered, the ergodic rate of Eve is given by
	\begin{equation}\label{ERE1}
	R_{E_i}=\frac{1}{\ln(2)}e^{-\frac{\lambda_e}{2}}\sum_{k=0}^{\infty}
	\frac{\lambda_e^k}{k!2^k\Gamma\left(k+\frac{1}{2}\right)}\sum_{l=0}^{u_1}\omega_{1,l}\mathcal{J}_{1,i}(x_{l}),
	\end{equation}
	where $x_{l}$ is the $l$-th root of the Laguerre polynomial $L_n(x)$, and the weight $\omega_{1,l}$ is given by
	\begin{equation}
	w_{1,l}={\frac {x_{l}}{\left(n+1\right)^{2}\left[L_{n+1}\left(x_{l}\right)\right]^{2}}},
	\end{equation}
	and
	\begin{equation}
	\mathcal{J}_{1,i}(x)=\frac{\Gamma\left(l+\frac{1}{2}\right)-\gamma\left(l+\frac{1}{2},\frac{t}{a_i}\right)}{1+x}e^x ,
	\end{equation}
	with $t=\frac{x}{2\rho_eN(1-\epsilon_e)d_1^{-\alpha_1}d_{E}^{-\alpha_E}}$.
\end{theorem}
\begin{proof}
	Please refer to Appendix C.
\end{proof}

\paragraph{\bf{Approach 1}}~
\begin{theorem}
	In the IRS-aided NOMA network considered, the approximate ASC of user 1 can be represented as
	\begin{equation}\label{ASC1}
	C_1\approx -\frac{1}{\ln(2)}e^{\frac{1}{a_1\rho d_{B_1}^{-\alpha_{B_1}}}}E_i\left(-\frac{1}{a_1\rho d_{B_1}^{-\alpha_{B_1}}}\right)-\frac{1}{\ln(2)}e^{-\frac{\lambda_e}{2}}\sum_{k=0}^{\infty}
	\frac{\lambda_e^k}{k!2^k\Gamma\left(k+\frac{1}{2}\right)}\sum_{l=0}^{u_1}\omega_{1,l}\mathcal{J}_{1,1}(x_{l}),
	\end{equation}
	where $E_i(\cdot)$ is the exponential integral.
\end{theorem}
\begin{proof}
	Based on \cite{yanyu2020}, the ergodic rate of user 1 is
	\begin{equation}\label{ERB1}
	R_{B_1}=-\frac{1}{\ln(2)}e^{\frac{1}{a_1\rho d_{B_1}^{-\alpha_{B_1}}}}E_i\left(-\frac{1}{a_1\rho d_{B_1}^{-\alpha_{B_1}}}\right).
	\end{equation}
	By Substituting \eqref{ERE1} and \eqref{ERB1} into \eqref{approximate ASC}, \eqref{ASC1} can be obtained. This completes the proof.
\end{proof}

\begin{theorem}
	In the IRS-aided NOMA network considered, the approximate ASC of user 2 can be represented as
	\begin{equation}\label{ASC2}
	\begin{aligned}
	C_2\approx&\log_2\left(1+\frac{a_2}{a_1}\right)-\frac{1}{\ln(2)}e^{-\frac{\lambda}{2}}\sum_{i=0}^{\infty}\frac{\lambda^i}{i!2^i\Gamma\left(i+\frac{1}{2}\right)}\sum_{p=0}^{u_2}\omega_p\mathcal{J}_2(t_p)\\
	&-\frac{1}{\ln(2)}e^{-\frac{\lambda_e}{2}}\sum_{k=0}^{\infty}
	\frac{\lambda_e^k}{k!2^k\Gamma\left(k+\frac{1}{2}\right)}\sum_{l=0}^{u_1}\omega_{1,l}\mathcal{J}_{1,2}(x_{l}),
	\end{aligned}
	\end{equation}
	where  $\omega_p=\frac{\pi}{u_2}$, $t_p=\cos\left(\frac{2p-1}{2u_2}\pi\right)$ and
	\begin{equation}
	\mathcal{J}_2(t)=\frac{\gamma\left(p+\frac{1}{2},\frac{t+1}{N(1-\epsilon)\rho a_1(1-t)d_1^{-\alpha_1}d_{B_2}^{-\alpha_{B_2}}}\right)}{1+t+\frac{2a_1}{a_2}}\sqrt{1-t^2}.
	\end{equation}
\end{theorem}
\begin{proof}
	Based on \cite{yanyu2020}, the ergodic rate of user 2 is
	\begin{equation}\label{ERB2}
	R_{B_2}\approx\log_2\left(1+\frac{a_2}{a_1}\right)-\frac{1}{\ln(2)}e^{-\frac{\lambda}{2}}\sum_{i=0}^{\infty}\frac{\lambda^i}{i!2^i\Gamma\left(i+\frac{1}{2}\right)}\sum_{p=0}^{u_2}\omega_p\mathcal{J}_2(t_p).
	\end{equation}
	By substituting \eqref{ERE1} and \eqref{ERB2} into \eqref{approximate ASC}, \eqref{ASC2} can be obtained. This completes the proof.
\end{proof}


\paragraph{\bf{Approach 2}}~

For the approximate ASC of user 2, we provide a more convenient approach.
It is worth noting that \eqref{A2} also produces highly accurate results, despite the computational complexity reduction compared to \eqref{ASC2}, as it will be demonstrated in Section \ref{NR}.
\begin{theorem}
	In the IRS-aided NOMA network considered, the approximate ASC of user 2 can be represented as
	\begin{equation}\label{A2}
	C_2^b\approx \log_2\left(\frac{\mu_2\rho d_1^{-\alpha_1}d_{B_1}^{-\alpha_{B_1}}+1}{(a_1\mu_2\rho d_1^{-\alpha_1}d_{B_1}^{-\alpha_{B_1}}+1)(1+a_2\mu d_1^{-\alpha_1}d_{E}^{-\alpha_{E}}\rho_e)}\right),
	\end{equation}
	 where $\mu_2$ is the expectation of $|\hat{h}_{B_2}|^2$, which can be obtained from \textbf{Lemma 3} in the case of $m_c=\min(m_1,m_2)$ and $m_d=\max(m_1,m_2)$.
\end{theorem}
\begin{proof}
	It can be seen from \eqref{approximate ASC} that the approximate ASC is defined as a difference function of the ergodic rate, while the ergodic rate is defined as a logarithmic function. Thus, we can define an approximate bound to the solution in \eqref{ERB} and \eqref{ERE} by invoking Jensen's inequality~\cite[\text{eq. (12.411)}]{gradshteyn2007}. Therefore, the ergodic rate of user 2 and Eve are
	\begin{equation}\label{bounded RB2}
	\begin{aligned}
	\mathbb{E}\big[\log_2\left(1+\gamma_{B_2}\right)\big]&\leq \log_2\left[1+\mathbb{E}(\gamma_{B_2})\right]\\
	&\leq \log_2\left(1+\frac{a_2\mu_2 d_1^{-\alpha_1}d_{B_1}^{-\alpha_{B_1}}}{a_1\mu_2 d_1^{-\alpha_1}d_{B_1}^{-\alpha_{B_1}}+\frac{1}{\rho}}\right),
	\end{aligned}
	\end{equation}
	and
	\begin{equation}\label{bounded RE2}
	\begin{aligned}
	\mathbb{E}\big[\log_2\left(1+\gamma_{E_2}\right)\big]&\leq \log_2\left[1+\mathbb{E}(\gamma_{E_2})\right]\\
	&=\log_2\left(1+a_2\mu d_1^{-\alpha_1}d_{E}^{-\alpha_{E}}\rho_e\right),
	\end{aligned}
	\end{equation}
	respectively.
	
	In this way, the approximate ASC can be expressed as a difference function of two approximate bounds. Then, upon substituting \eqref{bounded RB2} and \eqref{bounded RE2} into \eqref{approximate ASC}, \eqref{A2} can be obtained. This completes the proof.
\end{proof}


\subsubsection{Asymptotic ASC}~

Again we consider $\rho\to\infty$. Based on this assumption, the asymptotic expression of the ASC is given by\cite{moualeu2019transmit}
\begin{equation}\label{hsnrasc}
C_i^{\infty}\approx R_{B_i}^{\infty}-R_{E_i}.
\end{equation}

To obtain it, the asymptotic expressions of user 1's ASC and the ceiling of user 2's ASC are derived in the following corollaries.
\begin{corollary} \label{asymptotic ASC}
	The asymptotic ASC of user 1 is given by
	\begin{equation}\label{high snr C1}
	C_1^{\infty}\approx \log_2\left(\frac{1}{a_1\rho d_{B_1}^{-\alpha_{B_1}}}\right)-\frac{E_c}{\ln(2)}-\frac{1}{\ln(2)}e^{-\frac{\lambda_e}{2}}\sum_{k=0}^{\infty}
	\frac{\lambda_e^k}{k!2^k\Gamma\left(k+\frac{1}{2}\right)}\sum_{l=0}^{u_1}\omega_{1,l}\mathcal{J}_{1,1}(x_{l}),
	\end{equation}
	where $E_c$ denotes the Euler constant.
\end{corollary}
\begin{proof}
	Based on \cite{yanyu2020}, the asymptotic ergodic rate of user 1 is given by
	\begin{equation}\label{R1HSNR}
	R_{B_1}^{\infty}=\log_2\left(\frac{1}{a_1\rho d_{B_1}^{-\alpha_{B_1}}}\right)-\frac{E_c}{\ln(2)}.
	\end{equation}
Then, by substituting \eqref{R1HSNR} into \eqref{hsnrasc}, \eqref{high snr C1} can be obtained.
	This completes the proof.
\end{proof}
\begin{corollary}\label{ceiling}
	 The ceiling of $C_2$ in the high-SNR regime is given by
	\begin{equation}\label{high snr C2}
	C_2^{\infty}\approx \log_2(1+\frac{a_2}{a_1})-\frac{1}{\ln(2)}e^{-\frac{\lambda_e}{2}}\sum_{k=0}^{\infty}
	\frac{\lambda_e^k}{k!2^k\Gamma\left(k+\frac{1}{2}\right)}\sum_{l=0}^{u_1}\omega_{1,l}\mathcal{J}_{1,2}(x_{l}),
	\end{equation}
	or
	\begin{equation}\label{high snr C2_b}
	C_2^{b,\infty}\approx \log_2\left(\frac{1}{a_1\left(1+a_2\mu\rho_e d_1^{-\alpha_1}d_E^{-\alpha_{E}}\right)}\right).
	\end{equation}
\end{corollary}
\begin{proof}
	Based on \cite{yanyu2020}, the asymptotic ergodic rate of user 1 is given by
	\begin{equation}\label{R2HSNR}
	R_{B_2}^{\infty}=\log_2\left(1+\frac{a_2}{a_1}\right).
	\end{equation}
	Then, by substituting \eqref{R2HSNR} into \eqref{hsnrasc}, \eqref{high snr C2} can be obtained. 
	Finally, by substituting $\rho\to\infty$ into \eqref{A2}, \eqref{high snr C2_b} can be obtained.
	This completes the proof.	
\end{proof}

\begin{remark}
	Both the ASCs of user 1 and user 2 are affected by the number of IRSs and by the Nakagami-$m$ fading parameters.
\end{remark}

To gain deep insights into the network's performance, the high-SNR slope is worth estimating. Therefore, we first express the high-SNR slope as
\begin{equation}\label{slope}
S=\lim\limits_{\rho\to\infty}\frac{C(\rho)}{\log_2(\rho)}.
\end{equation}

In this analysis, we also consider $\rho\to\infty$, and maintain the consideration of arbitrary values of the average SNR of Eve's channel.

\begin{proposition}\label{high-SNR slope}
	By substituting \eqref{high snr C1} into \eqref{slope}, the high-SNR slope of the user 1 can be illustrated as 
	\begin{equation}
	S_1=1.
	\end{equation}
Then restituting \eqref{high snr C2} or \eqref{high snr C2_b} into \eqref{slope}, the high-SNR slope of the user 2 can be illustrated as
	\begin{equation}
	S_2=0.
	\end{equation}
\end{proposition}
\begin{remark}
Both the high-SNR slopes of user 1 and user 2 remain unaffected by the number of IRSs and by the Nakagami-$m$ fading parameters.
\end{remark}

Having completed the analyses of the IRS-aided network, all results related to both the secrecy diversity order and to the high-SNR slope are summarized in Table \ref{Secrecy diversity order and high-SNR slope for the IRS-aided network}. The orthogonal multiple access (OMA) benchmark schemes are described in \cite{yanyu2020}.

		

\begin{table}
	\renewcommand\arraystretch{1.5}
	\caption{\\Secrecy diversity order and high-SNR slope for the IRS-aided network}
	\centering
	\begin{tabular}{|c|c|c|c|}
		\hline
		Multiple-access scheme & user & $d$ & $S$ \\
		\hline
		\multirow{2}{*}{NOMA}
		& Bob1 & 1 & 1 \\
		\cline{2-4}
		& Bob2 & $m_sN$ & 0 \\
		\hline
		\multirow{2}{*}{OMA}
		& Bob1 & 1 & $\frac{1}{2}$ \\
		\cline{2-4}
		& Bob2 & $m_sN$ & $\frac{1}{2}$ \\
		\hline
	\end{tabular}
	\label{Secrecy diversity order and high-SNR slope for the IRS-aided network}
\end{table}

%

\section{Numerical Results}\label{NR}
In this section, our numerical results are presented for characterizing the performance of the network considered, complemented by our Monte-Carlo simulations to verify the accuracy attained. It is assumed that the power allocation coefficients of NOMA are $a_1=0.2$, $a_2=0.8$, respectively. 
The bandwidth of the DL is set to $BW = 1$ MHz, and the power of the AWGN is set to $\sigma^2 = -174 + 10\log_{10}(BW)$ dBm. 
In addition, the amplitude reflection coefficients of IRSs are set to $1$. The fading parameters are set to $m_1 =3$, $m_2 = m_3 = 1$, while the distance between the BS and IRS is set to $d_1 = 100$m. The IRS to user 2 and Eve distances are set to $d_{B2} = 10$m and $d_E = 50$m, and that of the BS-user 1 link is set to $d_{B1}= 20$m. The path loss exponents of the reflected links (i.e., BS-IRS, IRS-user 2 and IRS-Eve) and the direct BS-user 1 link are set to $\alpha_1=\alpha_{B2}=\alpha_{E} = 2.5$ and $\alpha_{B1} = 3.5$, respectively, unless otherwise stated.

\subsection{Secrecy outage probability}
For comparisons, we regard the IRS-aided OMA network as the benchmark. Specifically, IRSs are employed for providing legitamate access service to user 2 as well as for the illegitimate access of Eve to the BS. In Fig. \ref{fig.2} to Fig. \ref{fig.5}, we investigate the SOP, when the targeted secrecy capacity of the paired NOMA users is assumed to be $R_1=R_2=100$Kbps, which corresponds to the scenario considered in Section \ref{section 4}.
\begin{figure} [t!]
	\centering
	\includegraphics[width= 4in]{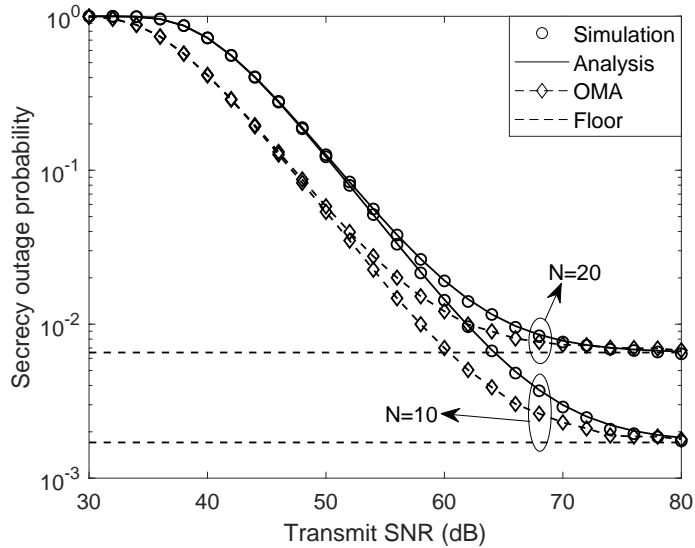}
	\caption{The SOP of user 1 versus the transmit SNR. The analytical results and the floor are calculated from \eqref{SOP1} and \eqref{SOP1 lower bound}.}
	\label{fig.2}
\end{figure}

Fig. \ref{fig.2} plots the SOP of user 1 versus the transmit SNR for different number of IRSs. It confirms the close agreement between the simulation and analytical results.
A specific observation is that the SOP of user 1 reduces upon reducing the number of IRSs, because the number of IRSs has no effect on the channel gain of user 1. By contrast, the channel gain of Eve increases, as the number of IRSs increases.
As a benchmark, the SOP curves of the IRS-aided OMA network are plotted for comparison. 
We observe that for user 1 in the IRS-aided OMA network has a better performance than that in the IRS-aided NOMA network in the high-SNR regime. 
This is because the transmit power allocated to user 1 in the NOMA network is lower than that in the OMA network due to the influence of the power allocation factor.
As the transmit SNR increases, we find that the SOP of user 1 tends to a constant, which is consistent with \textbf{Proposition \ref{SOP1 floor}}.
\begin{figure} [t!]
	\centering
	\includegraphics[width= 4in]{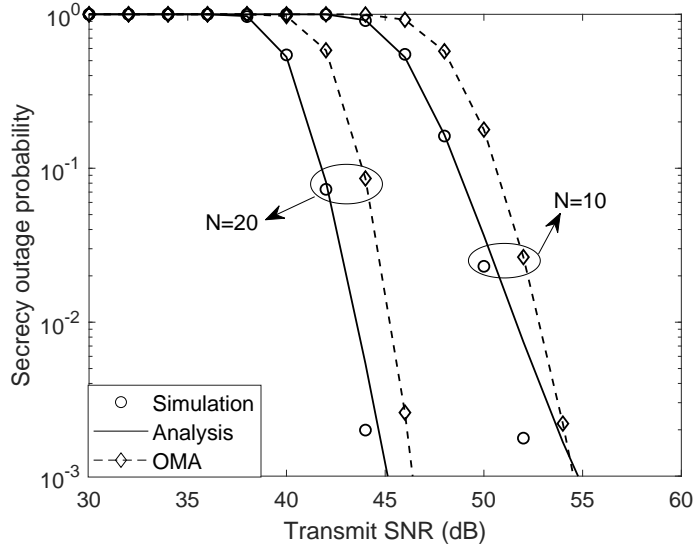}
	\caption{The SOP of user 2 versus the transmit SNR. The analytical results are calculated from~\eqref{SOP2 low-SNR}.}
	\label{fig.3}
\end{figure}

Fig. \ref{fig.3} plots the SOP of user 2 versus the transmit SNR. We observe that, owing to the central limit theorem (CLT) of the channel statistics of user 2, the analytical results are accurate in the low-SNR regime, but inaccurate in the high-SNR regime. 
As a benchmark, the SOP curves of the IRS-aided OMA network are also plotted for comparison. 
We observe that the performance of user 2 in the IRS-aided NOMA network is better than that of the IRS-aided OMA network. This is because the transmit power allocated to user 2 in the NOMA network is higher than that in the OMA network due to the influence of the power allocation factor.

\begin{figure} [t!]
	\centering
	\includegraphics[width= 4in]{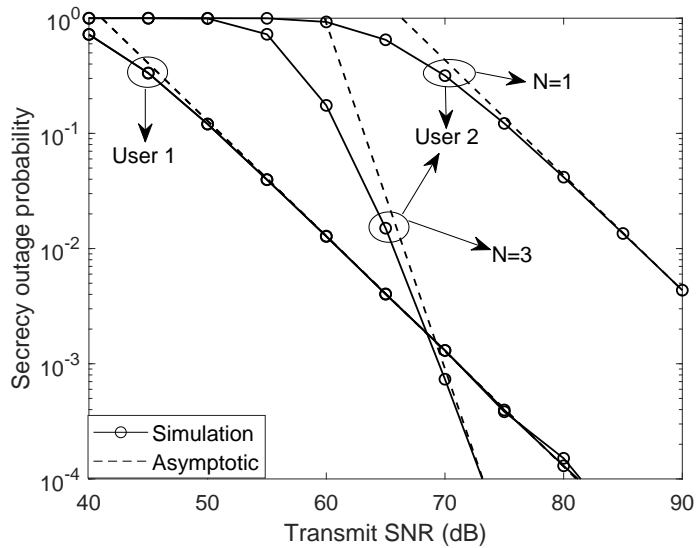}
	\caption{Asymptotic SOP versus the transmit power results in the case of $\rho_e=10$ dB. The asymptotic results are calculated from \eqref{asymptotic SOP Bob1} and \eqref{asymptotic SOP2}.}
	\label{fig.4}
\end{figure}
Since the SOP of user 2 in the high-SNR regime is not accurate in Fig.~\ref{fig.3}, we further plot the high-SNR asymptotic curves in the cases of $N = 1$ and $N = 3$ in Fig. \ref{fig.4}. We observe that the SOPs of user 1 and user 2 gradually approach their respective asymptotic curves, which validates our analysis. 
Furthermore, we also observe that in the cases of $N=1$ and $N=3$, the secrecy diversity orders of user 1 are both $1$ and the secrecy diversity orders of user 2 are $1$ and $3$, respectively, which is consistent with \textbf{Remark \ref{sdo1}} and \textbf{Remark \ref{sdo2}}.

\begin{figure} [t!]
	\centering
	\includegraphics[width= 4in]{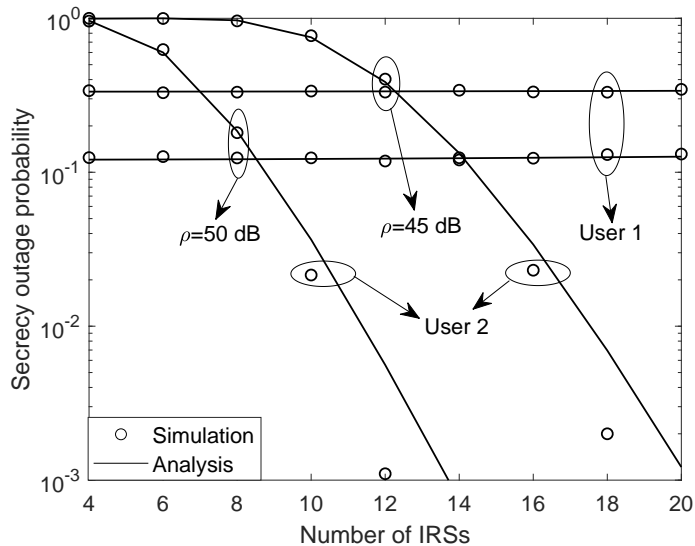}
	\caption{The SOP versus the number of IRSs. The analytical results are calculated from \eqref{SOP1} and~\eqref{SOP2 low-SNR}.}
	\label{fig.5}
\end{figure}
In Fig. \ref{fig.5}, the SOP curves versus the number of IRSs are depicted. 
We observe that, since we have global CSI for user 1, the SOP of user 1 is accurate.
On the other hand, the SOP of user 2 is accurate in the low-SNR regime. However, the SOP of user 2 is inaccurate in the high-SNR regime, owing to using the CLT-based channel statistics of user 2.
We also observe that the SOP of user 1 increases as the number of IRSs increases since the ergodic rate of user 1 is not affected by the number of IRSs. By contrast, the SOP of user 2 decreases as the number of IRSs increases, since the IRS-aided transmissions to Eve experience more severe path loss then that destined for user 2. 

\subsection{Average Secrecy Capacity}
In this subsection, the number of points for the Chebyshev-Gauss and Gauss-Laguerre quadratures are set to $u_1=u_2=100$. In Fig. \ref{fig.6} and Fig. \ref{fig.8}, we investigate the ASC of the IRS-aided NOMA network, which corresponds to the scenario considered in Section \ref{section 5}.

\begin{figure} [t!]
	\centering
	\includegraphics[width= 4in]{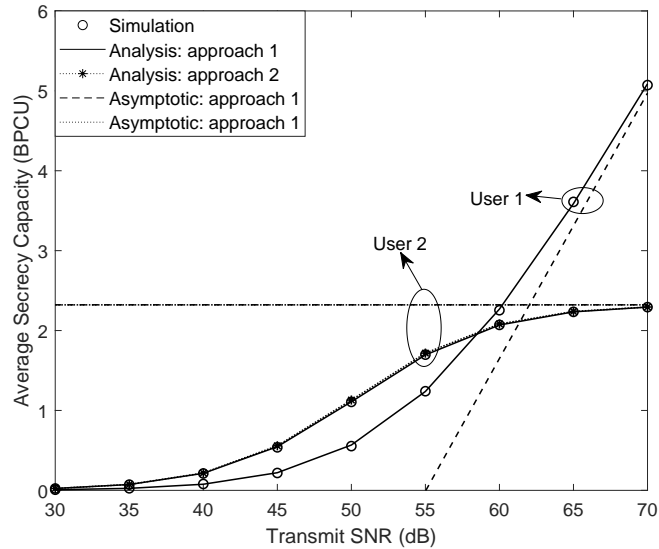}
	\caption{The ASC versus the transmit SNR in the case of N=30 and $\rho_e=30$ dB. The analytical results are calculated from \eqref{ASC1}, \eqref{ASC2} and \eqref{A2}, respectively. The asymptotic curves are calculated from \eqref{high snr C1}, \eqref{high snr C2} and \eqref{high snr C2_b}, respectively.}
	\label{fig.6}
\end{figure}

In Fig. \ref{fig.6}, the ASC curves for the IRS-aided NOMA network are depicted. On the one hand, for user 1, we observe that the simulation results match well with the analytical results, and the asymptotic results derived are accurate. On the other hand, for user 2, we observe that the analytical results both of approaches are highly accurate. Furthermore, the asymptotic results derived both from \eqref{high snr C2} and \eqref{high snr C2_b} are also accurate. Additionally, we also observe that the high-SNR slope of user 1 is $1$, while the ASC of user 2  approaches a ceiling in the high-SNR regime, which coincides with \textbf{Proposition \ref{high-SNR slope}}.



\begin{figure} [t!]
	\centering
	\includegraphics[width= 4in]{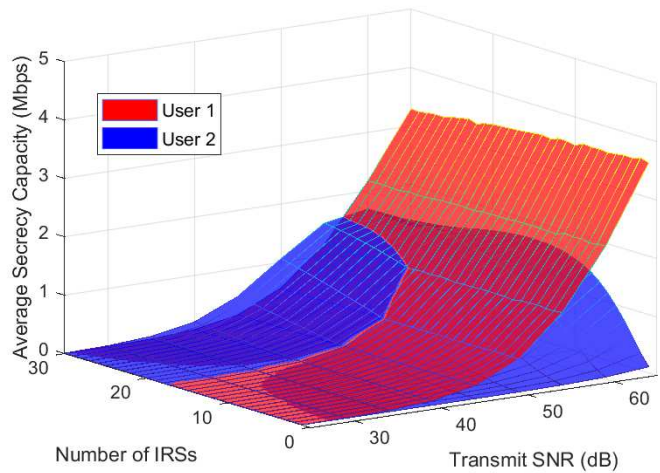}
	\caption{The ASC versus the numbers of IRS and the transmit SNR.}
	\label{fig.8}
\end{figure}
Fig. \ref{fig.8} plots the ASC of an IRS-aided NOMA network versus the number of IRSs and the transmit SNR. We assume that the transmit SNR of the paired NOMA users and Eve are identical. 
For user 1, we can observe that the ASC increases as the transmit SNR increases, since the BS-IRS-Eve links experience more servere path loss then the BS-user 1 link. We can also observe the ASC of user 1 decreases upon increasing the number of IRSs, since the ergodic rate of user 1 is not affected by the number of IRSs, while the ergodic rate of Eve increases upon increasong the number of IRSs.
For user 2, we observe that the ASC increases as the number of IRSs as well as the transmit SNR increase up to the ceiling given by \textbf{Corollary~\ref{asymptotic ASC}} and then decreases as the number of IRSs or the transmit SNR increase. This is because Eve has strong eavesdropping capability, which leads to an ergodic rate increase of Eve as increase the transmit SNR and the number of IRSs is.

\section{Conclusions}\label{CONCLUSIONS}
In this paper, the secrecy performance of IRS-aided NOMA networks was studied. 
Specifically, we first derived new channel gain expressions for the reflected links. Then, based on the new channel statistics, the closed-form SOP and ASC expressions were derived. 
Numerical results were presented for validating our results. 
Furthermore, the secrecy diversity orders and the high-SNR slopes have also been determined. 
The presence of the direct link between the BS and the cell-edge user as well as Eve deserves further investigation in our future research.

\numberwithin{equation}{section}
\section*{Appendix~A: Proof of \textbf{Lemma~\ref{lemma 2}} } \label{Appendix:A}

\renewcommand{\theequation}{A.\arabic{equation}}
\setcounter{equation}{0}

Firstly, according to \cite{yanyu2020}, the CDF of $X=\frac{\left(\sum_{n=1}^{N}|g_{2,n}||h_n|\right)^2}{N(1-\epsilon)}$ in the low-SNR regime is given by
\begin{equation}\label{CDFX}
\begin{aligned}
F_X(x)&=1-Q_{\frac{1}{2}}\left(\sqrt{\lambda},\sqrt{x}\right)\\
&=e^{-\frac{\lambda}{2}}\sum_{i=0}^{\infty}\frac{\lambda^i\gamma\left(i+\frac{1}{2},\frac{x}{2}\right)}{i!2^i\Gamma\left(i+\frac{1}{2}\right)}.
\end{aligned}
\end{equation}

Hence, the CDF of $\gamma_{B2}$ in the low-SNR regime can be derived as
\begin{equation}\label{CDFgamma21}
\begin{aligned}
F_{\gamma_{B2}}(x)&=\mathbb{P}\left(\gamma_{B2}<x\right)\\
&=F_X\left(\frac{x}{N(1-\epsilon)\rho(a_2-a_1x)d_1^{-\alpha_1}d_{B2}^{-\alpha_{B2}}}\right)\\
&=1-Q_{\frac{1}{2}}\left(\sqrt{\lambda},\sqrt{\frac{x}{N(1-\epsilon)\rho(a_2-a_1x)d_1^{-\alpha_1}d_{B2}^{-\alpha_{B2}}}}\right).
\end{aligned}
\end{equation}

Then, according to \cite{yanyu2020}, the CDF of $Y=\sum_{n=1}^{N}|g_{2,n}||h_n|$ in the high-SNR regime is given by
\begin{equation}
\begin{aligned}
F_Y(y)=\frac{m^N(4m_sm_l)^{-m_sN}}{\Gamma(2m_sN)}\gamma\left(2m_sN,2\sqrt{m_sm_l}y\right).
\end{aligned}
\end{equation}

Therefore, the CDF of $\gamma_{B2}$ in the high-SNR regime can be formulated as
\begin{equation}\label{CDFgamma21}
\begin{aligned}
F_{\gamma_{B2}}^{0_+}(x)&=\mathbb{P}\left(\gamma_{B2}<x\right)\\
&=F_Y\left(\sqrt{\frac{x}{\rho(a_2-a_1x)d_1^{-\alpha_1}d_{B2}^{-\alpha_{B2}}}}\right)\\
&=\frac{m^N(4m_sm_l)^{-m_sN}}{\Gamma(2m_sN)}\gamma\left(2m_sN,\frac{2\sqrt{m_sm_lx}}{\sqrt{\rho(a_2-a_1x)d_1^{-\alpha_1}d_{B2}^{-\alpha_{B2}}}}\right).
\end{aligned}
\end{equation}

This completes the proof.

\section*{Appendix~B: Proof of \textbf{Lemma~\ref{lemma 3}} } \label{Appendix:B}
\renewcommand{\theequation}{B.\arabic{equation}}
\setcounter{equation}{0}
By stipulating that $z_n=|g_{q,n}||h_n|$, and that $f_{z_n}$ is the probability density function (PDF) of $z_n$, according to \cite{yanyu2020}, the Laplace transform of $f_{z_n}$ is given by
\begin{equation}
\mathcal{L}_{f_{z_n}}(s)=\omega(s+2\sqrt{m_cm_d})^{-2m_c}{}_{2}F_{1}\left(2m_c,m_c-m_d+\frac{1}{2};m_c+m_d+\frac{1}{2};\frac{s-2\sqrt{m_cm_d}}{s+2\sqrt{m_cm_d}}\right).
\end{equation}

Assuming that $f_{Z}$ is the PDF of $Z$, the Laplace transform of $f_{Z}$ is given by
\begin{equation}\label{LTpdf}
\mathcal{L}_{f_{Z}}(s)=\underbrace{\omega^N(s+d)^{-aN}}_{f(s)}\underbrace{\left({}_{2}F_{1}\left(a,b;c;\frac{s-d}{s+d}\right)\right)^N}_{g(s)}.
\end{equation}

According to the relationship between the Laplace transform and moments, we have
\begin{equation}
\begin{aligned}
\mathbb{E}(Z^2)=\mathcal{L}_{f_{Z}}^{''}(0).
\end{aligned}
\end{equation}	
From \eqref{LTpdf}, we have
\begin{equation}
\mathcal{L}_{f_{Z}}^{'}(s)=J_1(s)+J_2(s)J_4(s),
\end{equation}
where
\begin{equation}
J_1(s)=-aN\omega^N(s+d)^{-aN-1}g(s),
\end{equation}
\begin{equation}
J_2(s)=f(s)\underbrace{N\left({}_{2}F_{1}\left(a,b;c;\frac{s-d}{s+d}\right)\right)^{N-1}}_{J_3(s)},
\end{equation}
and
\begin{equation}
J_4(s)=\frac{ab}{c}{}_{2}F_{1}\left(a+1,b+1;c+1;\frac{s-d}{s+d}\right)\frac{2d}{(s+d)^2}.
\end{equation}
Furthermore, we have
\begin{equation}\label{FORE}
\mathcal{L}_{f_{Z}}^{''}(s)=J_1^{'}(s)+J_2^{'}(s)J_4(s)+J_2(s)J_4^{'}(s),
\end{equation}
where
\begin{equation}\label{J1'}
\begin{aligned}
J_1^{'}(s)=aN(aN+1)\omega^N(s+d)^{-aN-2}g(s)-aN\omega^N(s+d)^{-aN-1}J_3(s)J_4(s),
\end{aligned}
\end{equation}	
\begin{equation}\label{J2'}
\begin{aligned}
J_2^{'}(s)=-aN\omega^N(s+d)^{-aN-1}J_3(s)+f(s)N(N-1)\left({}_{2}F_{1}\left(a,b;c;\frac{s-d}{s+d}\right)\right)^{N-2}J_4(s),
\end{aligned}
\end{equation}	
\begin{equation}\label{J4'}
\begin{aligned}
J_4^{'}(s)&=\frac{2abd}{c}(-2)(s+d)^{-3}{}_{2}F_{1}\left(a+1,b+1;c+1;\frac{s-d}{s+d}\right)+\\
&\frac{2abd}{c(s+d)^2}\frac{(a+1)(b+1)}{c+1}{}_{2}F_{1}\left(a+2,b+2;c+2;\frac{s-d}{s+d}\right)\frac{2d}{(s+d)^2}.
\end{aligned}
\end{equation}

By substituting $s=0$ into $J_2$, $J_4$, \eqref{J1'}, \eqref{J2'} and  \eqref{J4'}, we have
\begin{equation}\label{b9}
J_2(0)=\omega^Nd^{-aN}Nk_1^{N},
\end{equation}
\begin{equation}\label{b10}
J_4(0)=\frac{2ab}{cd}k_2,
\end{equation}
\begin{equation}\label{b11}
\begin{aligned}
J_1^{'}(0)
&=aN\omega^Nd^{-aN-2}k_1^{N-1}\left((aN+1)k_1-\frac{2ab}{c}Nk_2\right),
\end{aligned}
\end{equation} 	
\begin{equation}\label{b12}
\begin{aligned}
J_2^{'}(0)
&=\omega^Nd^{-aN-1}Nk_1^{N-2}\left(\frac{2ab}{c}(N-1)k_2-aNk_1\right),
\end{aligned}
\end{equation}		
\begin{equation}\label{b13}
\begin{aligned}
J_4^{'}(0)
&=\frac{4a(a+1)b(b+1)}{c(c+1)d^2}k_3-\frac{4ab}{cd^2}k_2.
\end{aligned}
\end{equation}
Then, by substituting \eqref{b9}-\eqref{b13} into \eqref{FORE}, and after some further mathematical manipulations, \eqref{E^2} can be obtained. This completes the proof.

\section*{Appendix~C: Proof of \textbf{Theorem~\ref{theorem 4}} } \label{Appendix:C}
\renewcommand{\theequation}{C.\arabic{equation}}
\setcounter{equation}{0}
Based on the above assumptions, the CDF of $\gamma_{E_i}$ in the low-SNR regime can be expressed as
\begin{equation}
\begin{aligned}
F_{\gamma_{E_i}}(x)&=F_X\left(\frac{x}{a_i\rho_eN(1-\epsilon_e)d_1^{-\alpha_1}d_{E}^{-\alpha_E}}\right)\\
&=1-Q_{\frac{1}{2}}\left(\sqrt{\lambda_e},\sqrt{\frac{x}{a_i\rho_eN(1-\epsilon_e)d_1^{-\alpha_1}d_{E}^{-\alpha_E}}}\right)\\
&=e^{-\frac{\lambda_e}{2}}\sum_{k=0}^{\infty}\frac{\lambda_e^k\gamma\left(k+\frac{1}{2},\frac{x}{2a_i\rho_eN(1-\epsilon_e)d_1^{-\alpha_1}d_{E}^{-\alpha_E}}\right)}{k!2^k\Gamma\left(k+\frac{1}{2}\right)},
\end{aligned}
\end{equation}
where we have $\epsilon_e=\frac{1}{m_1m_3}\left(\frac{\Gamma(m_1+\frac{1}{2})}{\Gamma(m_1)}\right)^2\left(\frac{\Gamma(m_3+\frac{1}{2})}{\Gamma(m_3)}\right)^2$ and $\lambda_e=\frac{N\epsilon_e}{1-\epsilon_e}$.

Therefore, the ergodic rate of Eve can be expressed as
\begin{equation}
\begin{aligned}
R_{E_i}&=\mathbb{E}\left(\log_2(1+\gamma_{E_i})\right)\\
&=\int_{0}^{\infty}\log_2(1+x)df_{\gamma_{E_i}}(x)\\
&=\int_{0}^{\infty}(1-F_{\gamma_{E_i}}(x))d\log_2(1+x)\\
&=\frac{1}{\ln(2)}e^{-\frac{\lambda_e}{2}}\sum_{k=0}^{\infty}
\frac{\lambda_e^k}{k!2^k\Gamma\left(k+\frac{1}{2}\right)}J_{5,i},
\end{aligned}
\end{equation}
where
\begin{equation}\label{c3}
J_{5,i}=\int_{0}^{\infty}\frac{\Gamma\left(k+\frac{1}{2}\right)-\gamma\left(k+\frac{1}{2},\frac{x}{2a_i\rho_eN(1-\epsilon_e)d_1^{-\alpha_1}d_{E}^{-\alpha_E}}\right)}{1+x}dx.
\end{equation}
Next, by using the Gauss-Laguerre quadrature, \eqref{c3} can be rewritten as
\begin{equation}
J_{5,i}\approx\sum_{l=0}^{n}\omega_{5,l}\mathcal{J}_5(x_{5,l}),
\end{equation}
where $x_{5,l}$ is the $l$-th root of the Laguerre polynomial $L_n(x)$, the weight $\omega_{5,l}$ is given by
\begin{equation}
w_{5,l}={\frac {x_{l}}{\left(n+1\right)^{2}\left[L_{n+1}\left(x_{l}\right)\right]^{2}}},
\end{equation}
and
\begin{equation}
\mathcal{J}_{5,i}(x)=\frac{\Gamma\left(l+\frac{1}{2}\right)-\gamma\left(l+\frac{1}{2},\frac{t}{a_i}\right)}{1+x}e^x,
\end{equation}
with $t=\frac{x}{2\rho_eN(1-\epsilon_e)d_1^{-\alpha_1}d_{E}^{-\alpha_E}}$. This completes the proof.


\bibliographystyle{IEEEtran}
\bibliography{mybib}

\end{document}